\documentclass[conference]{IEEEtran}
\IEEEoverridecommandlockouts
\usepackage{cite}
\usepackage{amsmath,amssymb,amsfonts, booktabs, amsthm, nccmath, caption}
\usepackage{algorithmic}
\usepackage{graphicx}
\usepackage{textcomp}
\usepackage{xcolor}
\usepackage[linesnumbered,ruled]{algorithm2e}
\PassOptionsToPackage{hyphens}{url}\usepackage{hyperref}
\usepackage{amsmath, tablefootnote, tikz, array}
\usetikzlibrary{patterns}

\newtheorem{lemma}{Lemma}

\newtheorem{problem}{Problem}

\newcommand{\hide}[1]{}

\newcommand*{\QEDB}{\null\nobreak\hfill\ensuremath{\square}}%

\SetKwFor{For}{for}{do}{end\ for for}%

\def\BibTeX{{\rm B\kern-.05em{\sc i\kern-.025em b}\kern-.08em
    T\kern-.1667em\lower.7ex\hbox{E}\kern-.125emX}}

\begin{document}
\allowdisplaybreaks
\raggedbottom

\title{REFORM: Fast and Adaptive Solution for \\Subteam Replacement}

\author{\IEEEauthorblockN{Zhaoheng Li, Xinyu Pi, Mingyuan Wu, Hanghang Tong}
\IEEEauthorblockA{\textit{Department of Computer Science} \\
\textit{University of Illinois at Urbana-Champaign}\\
Urbana, United States \\
\{zl20, xinyupi2, mw34, htong\}@illinois.edu}
}

\maketitle

\begin{abstract} 
\,{\em Subteam Replacement}: given a team of people embedded in a social network to complete a certain task, and a subset of members (i.e., subteam) in this team which have become unavailable, find another set of people who can perform the subteam’s role in the larger team. We conjecture that a good candidate subteam should have high skill and structural similarity with the replaced subteam while sharing a similar connection with the larger team as a whole. Based on this conjecture, we propose a novel graph kernel which evaluates the goodness of candidate subteams in this holistic way freely adjustable to the need of the situation. To tackle the significant computational difficulties, we equip our kernel with a fast approximation algorithm which (a) employs effective pruning strategies, (b) exploits the similarity between candidate team structures to reduce kernel computations, and (c) features a solid theoretical bound on the quality of the obtained solution. We extensively test our solution on both synthetic and real datasets to demonstrate its effectiveness and efficiency. Our proposed graph kernel outputs more human-agreeable recommendations compared to metrics used in previous work, and our algorithm consistently outperforms alternative choices by finding near-optimal solutions while scaling linearly with the size of the replaced subteam.
\end{abstract}


\section{Introduction}
Replacing a specific team member in a team is a common problem across many application domains which can arise in many situations, whether it is to improve team composition, address a change of role of the team, or simply because the team member has become unavailable. In environments with highly dynamic team structures, it is often the case that a larger scale of this problem needs to be addressed: what if \textit{multiple} team members need to be replaced at the same time? Representative examples include the US movie industry with its long-standing practice of replacing starring actors/actresses during long productions to evaluate their value\cite{imdbexample}, and academia with its steadily growing project team size leading to increased probability of members being replaced sometime during the project\cite{academiaexample}.

Based on this observation, we introduce a novel yet challenging problem, namely \textit{Subteam Replacement}: given (1) a large social network modeling the skills of individuals and their familiarity with working with each other, (2) a team of people within this network formed to complete some specific task, and (3) a subset of team members (i.e., \textit{Subteam}) which 
has become unavailable, we aim to construct the best replacement subteam from the social network to fill in the role of the replaced subset of team members.

Many characteristics of effective teamwork have been identified and studied in organizational management literature. Most related to our problem are that team members prefer to work with people they have worked
with before\cite{hinds2000choosing} and that teams with familiarity between members tend to have better performance\cite{familiarTeam}. This naturally leads to the idea that a good candidate subteam should have good structural similarity, preserving the links of the replaced subteam\cite{TeamReplacement}.

Additionally, to preserve the ability of the team to perform the assigned task, the replacement subteam should exhibit skill similarity, having a set of skills similar to the replaced subteam. Taking inspiration from the work in \cite{AtypicalCombinations} which points out that atypical combinations of ideas is a major factor of success in scientific research, and other recent research\cite{haas2016secrets} identifying the balance of skills within a team to be a key indicator of success rather than individual skill, we also propose that during the evaluation of the candidate subteam, instead of matching the properties of individuals, we instead aim to match the \textit{interaction} between individuals. 

It is common practice in fields such as business analytics to model a social network of individuals as a graph\cite{bonchi2011social}. Given the rich history of using kernel-based methods to measure graph similarity \cite{GraphKernelSurvey}, we propose an adaptation of the random walk graph kernel \cite{SvnKernel} for edge-labeled graphs as a suitable way to measure the goodness of candidate replacement subteams according to our observations: it holistically evaluates subteams in terms of skill and structure both locally and globally from the perspective of human interactions.

However, \textit{Subteam Replacement} is computationally hard. A simple reduction from best-effort subgraph search on weighted, attributed graphs \cite{GraphSearchMultiAttribute} shows that the problem of finding the best replacement is in fact NP-Complete, and brute-forcing every possible subteam is necessary to find the optimal solution. For a social network with $n$ members and a subteam to be replaced of size $s$ will require $O(_nC_s) = O( \frac{n!}{s!(n-s)!})$ kernel computations, which is computationally unfeasible for all but the most trivial problem settings. For example, given a network of size $n = 70$ and a subteam of size $s = 4$ out of a subgraph of size 9 as the team, we found that it can still take upwards of 2 hours to find the replacement using brute force.

To address this issue, it is necessary to both (1) speedup individual graph kernel computations, and (2) reduce the number of subteams we evaluate. We design a fast approximation algorithm that exploits similarity between candidate subteams to speedup subsequent kernel computations to address the former, and only evaluating a few of the most promising subteams to return a near-optimal solution to address the latter. Through comprehensive experiments, we show that our algorithm (1) has a strong theoretical bound on the quality of the solution it returns, (2) recommends subteams highly correlated with human perception, and (3) scales linearly with the size of the replaced subteam making it possible to perform Subteam Replacement on complicated large-scale team structures.

The main contributions of the paper are as follows:\begin{itemize}
    \item \textbf{Problem Definition}. We provide the formal definition for \textit{Subteam Replacement}. 
    \item \textbf{Adaptive Similarity Measure}. We propose a novel graph kernel as our similarity measure between teams adjustable to different circumstances the team is formed for.
    \item \textbf{Algorithm and Analysis}. We propose \textsc{REFORM}, an efficient algorithm for solving \textit{Subteam Replacement} addressing each of its computational challenges, and derive a solid theoretical bound for the quality of the solution obtained from \textsc{REFORM}.
    \item \textbf{Experiments}. We demonstrate the efficacy of \textsc{REFORM} on various real-world datasets using both quantitative and user studies. 
\end{itemize}

\begin{figure}[htbp]
\centerline{\includegraphics[width=\linewidth]{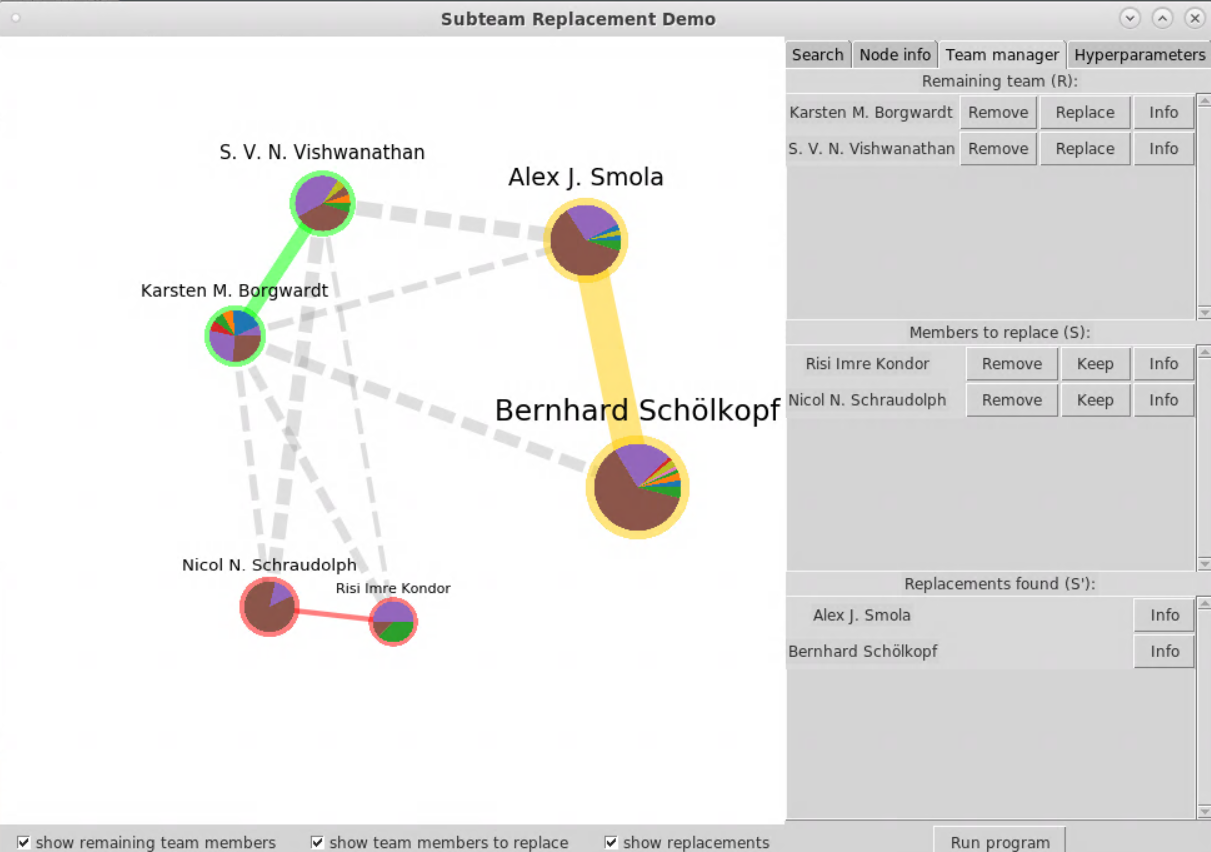}}
\captionsetup{labelformat=empty}
\caption{\textbf{Fig. 1.} Example of Subteam Replacement. Green nodes are remaining team members, red nodes are the replaced team members, and yellow nodes are the found replacements.}
\end{figure}

\section{Problem definition}
\begin{table}[ht]
\caption{Table of Symbols}

\begin{tabular}{ll}
\toprule
Symbols              & Definition                            \\ \midrule
$G = (V, E)$         & The social network                    \\
$A_{n\times n}$            & The adjacency matrix of $G$           \\
$L_{n\times l}$            & The skill indicator matrix            \\
$W_{l\times l}$            & The skill pair relevance matrix            \\ \hline
$T \subset V$        & The given team                        \\
$S \subset T$        & The set of team members to replace    \\
$R = T \setminus S$         & The remaining team after removing $S$ \\ \hline
$G_T = \{A[T,T], L[T,:]\}$ & The subgraph of G indexed by T        \\
$A_T = A[T,T]$ & The submatrix of A indexed by T       \\
$L_T = L[T,:]$ & The submatrix of L indexed by T       \\ \hline
$l$ & The total number of skills       \\
$t$ & The team size ($|T| = t$)      \\
$s$ & The number of people to replace ($|S| = s$)      \\
$n$ & The total number of individuals in $A$      \\ 
\bottomrule
\end{tabular}
\end{table}
Table 1 lists the symbols that we will use throughout the paper. 
We describe the input data to \textit{Subteam Replacement} as follows:
\begin{enumerate}
\item A social network as $n$ individuals organized in a weighted, labeled graph $G = (V, E)$ with non-negative edge weights, where each vertex corresponds to one individual and each edge represents a relationship between two individuals. As the labels are an integral part of the graph, we will also refer to $G$ as $G = \{A, L\}$.  
    The graph structure will be represented by $G$'s non-negative $n\times n$ adjacency matrix $A$. Edge weights correspond to the strength of relationships; a heavily weighted edge represents two individuals are familiar with working with each other and thus have a desirable low communication cost between them.
    The node labels will be represented by the non-negative $n\times l$ skill matrix $L$. Each $i^\textrm{th}$ row vector represents the $i^\textrm{th}$ individual's proficiency in each of the $l$ skills; for example, if the set of skills $l = \{HTML/CSS, C\!+\!+, SQL\}$, then an individual having the skill vector $[0, 1, 1]$ means that they have skill in $C\!+\!+$ and $SQL$ but no skill in $HTML/CSS$.
\item A $l\times l$ non-negative upper triangular matrix $W$, where each entry $W[i,j]$, $i\! \leq\! j$ represents the relevance of the skill pair $i$ and $j$ in completing the assigned task. For example, if the skills are $\{HTML/CSS, C\!+\!+, SQL\}$, having $W = \{\!\{0, 1, 0\}, \{0, 0, 1\}, \{0, 0, 0\}\!\}$ will mean that the skill pairs of (1) $HTML/CSS$ and $C\!+\!+$, (2) $SQL$ and $C++$ are relevant and complementary in completing the task.

\item A team of individuals $T \subsetneq V$, and the subset of members to replace $S \subsetneq T$. We will also denote the remaining team members as $R = T \setminus S$, and we would like to impose the following two size constraints, including 
(1) $1 \leq |S| \leq |T| - 1$, and 
(2) $n - |T| \geq |S|$.\footnote{This ensures we have enough candidates to form a replacement subteam $S'$ given $|S'| = |S|$.}
\end{enumerate}

We refer to matrix indexing using NumPy convention, i.e.  $A[T, T]$ is the submatrix of $A$ formed by the rows and columns indexed by elements in $T$, and $A[T,:]$ is the submatrix of $A$ formed by the rows indexed by elements in $T$. Given the team $T$ and our social network $G$, the team's network can be represented as a subgraph of $G$, namely $G(T) = \{A[T, T], L[T,:]\}$.

If we replace $S$ with another subteam $S'$, the new team can be represented by the subgraph $G(R\cup S') = \{A[R\cup S', R\cup S'], T[R\cup S',:]\}$. We will be studying the case where $|S'| = |S|$ in our paper. For the sake of clarity, we use the following notation to represent these elements in this paper:\begin{itemize}
    \item $A_T = A[T, T]$
    \item $L_T = L[T, :]$
    \item $G_T = G(T) = \{A[T, T], L[T, :]\}$
\end{itemize}

With the above notations and assumptions, we formally define our problem as follows:
\begin{problem}{Subteam Replacement}
\label{prof:subteam:replacement}
\begin{description}
\item[Given:] (1) a labelled social network $G = (V, E) = \{A, L\}$, (2) a skill pair relevance matrix $W$, 
(3) a team of individuals $T \subsetneq V$, and (4) a subteam of individuals to be replaced $S \subsetneq T$;
\item[Find:] The best replacement subteam $S' \in (T \setminus S)$ to fill in the role of the replaced subteam $S$ in $T$.
\end{description}
\end{problem}

\section{Proposed Methodology}
In this section, we present our solution to \textit{Subteam Replacement}. We first introduce our goodness measure - the random walk graph kernel for edge-labeled graphs - that evaluates subteams according to our objectives. Next, we introduce our speedup algorithm to efficiently replace 1 member, then build upon this algorithm to obtain our greedy approximation algorithm for simultaneous multiple member replacement. We also discuss the implications of the mathematical properties of our algorithm.

\subsection{Proposed Subteam Goodness Measure}

As mentioned in the introduction, the random walk graph kernel for edge-labeled graphs is a fitting goodness measure that can evaluate subteams according to the identified factors contributing to a successful team. Presented below is the basic form of an edge-labeled random walk graph kernel\cite{SvnKernel}:

\begin{equation}
\textnormal{Ker}(G_1, G_2) = y^T(I-c(E_x\odot A_x))^{-1}x
\end{equation}

Where $\otimes$ and $\odot$ are respectively the Kronecker and Hadamard (elementwise) product between 2 matrices, $y = y_1\otimes y_2$ and $x = x_1\otimes x_2$ are uniform starting and stopping vectors, $A_x = A_1 \otimes A_2$ is the adjacency matrix of the product graph of $G_1$ and $G_2$, and $E_x\! =\! \sum^{k}_{i=1}E_1(:,:,i)\!\otimes\! E_2(:,:,i)$ is the edge attribute similarity matrix for the two graphs defined using the tensors $E_1$ and $E_2$ containing the $k$ attributes for each edge in $G_1$ and $G_2$.


We next customize the kernel by defining each slice in our edge attribute tensor in terms of the skill matrix $L$:

\begin{equation}
\medmath{E(:,:,i+(j-1)l) = \textrm{max}(L[:,i]L[:,j]^T, L[:,j]L[:,i]^T)}
\end{equation}

We define the $\textrm{max}$ operation in this case to be taking the elementwise maximum of the argument matrices:

\begin{equation}
\textrm{max}(A_1, A_2)_{ij} = \textrm{max}(A_{1_{ij}}, A_{2_{ij}}) \forall i,j
\end{equation}

Adhering to the idea of matching skill interaction, each slice in the tensor measures the extent that each pair of individuals can utilize a pair of skills. For example, for every cell $[a, b]$ in equation 2, the value is a measurement of how well $a$ and $b$ can combine skills $i$ and $j$: either $a$ performs $i$ and $b$ performs $j$ or vice versa, whichever is more desirable, hence the $\textrm{max}$.

Putting everything together, we substitute the slice-wise tensor product in the above edge-labeled graph kernel (equation 1) with $W$ and $L$ to arrive at our proposed kernel measure:

\begin{multline}
\medmath{E_x = \sum^{l}_{i=1}\sum^{l}_{j=i} W[i,j] * (\textrm{max}(L_1[:,i]L_1[:,j]^T, L_1[:,j]L_1[:,i]^T)}\\ 
\medmath{\otimes (\textrm{max}(L_2[:,i]L_2[:,j]^T, L_2[:,j]L_2[:,i]^T)}
\end{multline}

We scale each skill pairing with $W$ representing how important the skill pairing is for the task at hand; note that setting $W$ as upper triangular effectively avoids double-counting pairs of distinct skills.

{\em Remark.} Our proposed edge-labeled graph kernel is similar to but bears subtle difference from the node-labelled graph kernel in \cite{TeamReplacement}. Specifically, we find that the node-labelled graph kernel might overly emphasize the individual's skills over their network connectivity, which could be detrimental in the process of pruning unpromising candidates. For example, a good candidate might be pruned when using the node-labelled kernel, while our proposed edge-labelled kernel will prevent such good candidates from being pruned as we will show in the next subsection.

\subsection{Single Member Replacement}

Finding the optimal solution for \textit{Subteam Replacement} is NP-Complete. A simple and yet effective way to approach an NP-Complete problem would be to use a greedy approximation algorithm picking the best candidate to add to the team at each step. To do this, we will first present an efficient algorithm for replacing 1 person.

Since we define the goodness of the team by the ability of its team members to interact with each other, we naturally want to avoid candidates which have no connections to the remaining team $R$. Therefore, pruning can be employed to achieve a considerable speedup while having minimal impact on the quality of the suggested candidate. With our kernel, we can prove mathematically that pruning has \textit{no} impact on the quality of the solution - we will never prune the optimal candidate.

Table 2 contains some recurring expressions which we use as shorthand in the following proof and sections thereafter.
\begin{table}[ht]
\caption{Table of Shorthand expressions}

\begin{tabular}{ll}
\toprule
Symbol              & Shorthand for                            \\ \midrule
$E_{\textrm{max}(i,j)}$         & $\textrm{max}(L[:,i]L[:,j]^T, L[:,j]L[:,i]^T)$                    \\
$E_{\alpha\times \beta}$           & $ \sum^l_{i=1}\sum^l_{j=i}W[i,j]*E_{\alpha_{ \textrm{max}(i,j)}}\otimes E_{\beta_{\textrm{max}(i,j)}}$            \\
$G_{\alpha \times \beta}$            & $c(E_{\alpha \times \beta} \odot (A_\alpha \otimes A_\beta))$            \\ \bottomrule
\end{tabular}
\end{table}
\begin{lemma}
\textsc{Validity of pruning.} Given any team $T$ and any person to replace $p \in T$, and 2 candidates $q, q'$ not in $T$, if $q$ is connected to at least 1 member in $T$ and $q'$ is not connected to any members in $T$, there is:

\begin{equation}
\textnormal{Ker}(G_T, G_{(T\setminus\{p\})\cup\{q\}}) \geq \textnormal{Ker}(G_T, G_{(T\setminus\{p\})\cup\{q'\}})
\end{equation}
\end{lemma}

\begin{proof} Let:\begin{itemize}
    \item $G_T = G_0 = \{A_0, L_0\}$
    \item $G_{(T\setminus\{p\})\cup\{q\}} = G_1 = \{A_1, L_1\}$
    \item $G_{(T\setminus\{p\})\cup\{q'\}} = G_2 = \{A_2, L_2\}$
    \item $E_{\textrm{max}(i,j)} = \textrm{max}(L[:,i]L[:,j]^T, L[:,j]L[:,i]^T)$
\end{itemize}
By Taylor expansion of equation 1, we have:

$\textnormal{Ker}(G_0, G_1) = y^T\sum^\infty_{k=0}c(E_{0\times1}\odot(A_0\otimes A_1))^kx$, where $E_{0\times1} = \sum^l_{i=1}\sum^l_{j=i}W[i,j]*E_{0_{\textrm{max}(i,j)}}\otimes E_{1_{\textrm{max}(i,j)}}$

$\textnormal{Ker}(G_0, G_1) = y^T\sum^\infty_{k=0}c(E_{0\times2}\odot(A_0\otimes A_2))^kx$, where $E_{0\times2} = \sum^l_{i=1}\sum^l_{j=i}W[i,j]*E_{0_{\textrm{max}(i,j)}}\otimes E_{2_{\textrm{max}(i,j)}}$

It is sufficient to show that $c(E_{0x1}\odot(A_0\otimes A_1))^k \geq c(E_{0x2}\odot(A_0\otimes A_2))^k$ for all k; we define the element-wise matrix inequality $A \geq B$ as $A \geq B \Longleftrightarrow A_{ij} \geq B_{ij} \forall i,j$.

\textsc{Proof by induction - base case}: $k = 0$ is trivial, therefore we will start with $k = 1$:

LHS: $c(E_{0x1}\odot(A_0\otimes A_1))$

$ = \sum^l_{i=1}\sum^l_{j=i}cW[i,j]\!*\!(E_{0_{\textrm{max}(i,j)}}\!\otimes\! E_{1_{\textrm{max}(i,j)}})\!\odot\!(A_0\!\otimes\! A_1)$

$ = \sum^l_{i=1}\sum^l_{j=i}cW[i,j]\!*\!(E_{0_{\textrm{max}(i,j)}}\!\odot\! A_0) \!\otimes\!(E_{1_{\textrm{max}(i,j)}}\!\odot\! A_1)$

$ \geq \sum^l_{i=1}\sum^l_{j=i}cW[i,j]\!*\!(E_{0_{\textrm{max}(i,j)}}\!\odot\! A_0) \!\otimes\!(E_{2_{\textrm{max}(i,j)}}\!\odot\! A_2)$

$= c(E_{0x2}\odot(A_0\otimes A_2))$ (RHS)

Where $(E_{1_{\textrm{max}(i,j)}}\odot A_1) \geq (E_{2_{\textrm{max}(i,j)}}\odot A_2) \forall i, j$ because $A_1$ has at least 1 nonzero element in the last row and column due to $q$ having at least 1 connection with $R$, while $A_2$ has only zeroes in its last row and column due to $q'$ having no connections with $R$. The first $t-1$ rows and columns of $A_1$ and $A_2$, $E_{1_{\textrm{max}(i,j)}}$, $E_{2_{\textrm{max}(i,j)}}$ are otherwise identical as they both pairwise represent the adjacency and skill interaction of the remaining team $R$.

\textsc{Proof by induction - inductive step}: Assuming $c^m(E_{0x1}\odot(A_0\otimes A_1))^m \geq c^m(E_{0x2}\odot(A_0\otimes A_2))^m \forall m \in \{0,...,k-1\}$:

$c^k(E_{0x1}\odot(A_0\otimes A_1))^k$

$\geq c^{k-1}(E_{0x1}\odot(A_0\otimes A_1))^{k-1}c(E_{0x2}\odot(A_0\otimes A_2))$

$\geq c^k(E_{0x2}\odot(A_0\otimes A_2))^k$\\
where we have both \\$c(E_{0x1}\odot(A_0\otimes A_1))\geq c(E_{0x2}\odot(A_0\otimes A_2))$ and\\ $c^{k-1}(E_{0x1}\odot(A_0\otimes A_1))^{k-1} \geq c^{k-1}(E_{0x2}\odot(A_0\otimes A_2))^{k-1}$ by the inductive assumption, which completes the proof. 

\end{proof}


We now present our speedup for computing the updated graph kernel for each candidate. We will be reusing the following notations from the previous lemma:\begin{itemize}
    \item $G_T = G_0 = \{A_0, L_0\}$ is the subgraph of the input team.
    \item $G_{(T\setminus\{p\})\cup\{q\}} = G_1 = \{A_1, L_1\}$ is the subgraph of the input team after removing $p$ and inserting some arbitrary candidate $q$.
\end{itemize}
    
Without loss of generality, we assume that $q$ is the last person in $T$; the last row of $L$ contains the skill vector for $q$.

Upon inspection, the input team $T$ and the candidate team $(T\setminus\{p\})\cup\{q\}$ differ only by one person. This means that the matrices $A_0$ and $A_1$ differ only in their last row and column, and the matrices $L_0$ and $L_1$ differ only in their last row.

Given these observations, we can rewrite $A_1$ and $L_1$ as follows:\begin{itemize}
    \item $A_1 = A_R + A_q$, where $A_R$ is $A_1$ with the last row and column zeroed out, and $A_q$ is $A_1$ with all but the last row and column zeroed out.
    \item $L_1 = L_R + L_q$, where $L_R$ is $L_1$ with the last row zeroed out, and $L_q$ is $L_1$ with all but the last row zeroed out.
\end{itemize}

We can now rewrite the graph kernel using our newly defined notations:

\begin{multline*}
\textnormal{Ker}(G_1, G_0) \\
\medmath{= y^T(I\!-\!c(\sum^l_{i=1}\sum^l_{j=i}W[i,j]\!*\!(E_{1_{\textrm{max}(i,j)}}\!\otimes\! E_{0_{\textrm{max}(i,j)}})\!\odot\!(A_1\!\otimes\! A_0))^{-1}x} \\
\medmath{= y^T(I\!-\!c(\sum^l_{i=1}\sum^l_{j=i}(E_{1_{\textrm{max}(i,j)}}\!\odot \!A_1)\!\otimes\!(\underbrace{W[i,j]\!*\!E_{0_{\textrm{max}(i,j)}}\!\odot\! A_0}_{\textbf{Z\textsubscript{ij} (invariant w.r.t. q)}})))^{-1}x} \\
\medmath{= y^T(I-c(\sum^l_{i=1}\sum^l_{j=i}(E_{1_{\textrm{max}(i,j)}}\odot A_1)\otimes Z_{ij}))^{-1}x}
\end{multline*}
Rewriting $A_1$ as $A_R + A_q$ and taking apart $E_{1_{\textrm{max}(i,j)}}$:

\begin{multline*}
\medmath{y^T(I-c(\sum^l_{i=1}\sum^l_{j=i}(\underbrace{E_{R_{\textrm{max}}(i,j)}\odot A_R}_{\textbf{Y\textsubscript{ij} (invariant w.r.t. q)}})\otimes Z_{ij})}\\
\medmath{-c(\sum^l_{i=1}\sum^l_{j=i}(\underbrace{\textrm{max}(L_R[:,i]L_q[:,j]^T, L_R[:,j]L_q[:,i]^T)\odot A_q}_{\textbf{B\textsubscript{ij} (depends on q)}})\otimes Z_{ij})}\\
\medmath{-c(\sum^l_{i=1}\sum^l_{j=i}(\underbrace{\textrm{max}(L_q[:,i]L_R[:,j]^T, L_q[:,j]L_R[:,i]^T)\odot A_q}_{\textbf{B\textsuperscript{T}\textsubscript{ij} (depends on q)}})\otimes Z_{ij})}\\
\medmath{-c(\sum^l_{i=1}\sum^l_{j=i}(\underbrace{\textrm{max}(L_R[:,i]L_q[:,j]^T, L_R[:,j]L_q[:,i]^T)\odot A_R}_{\textbf{0}})\otimes Z_{ij})}\\
\medmath{-c(\sum^l_{i=1}\sum^l_{j=i}(\underbrace{E_{q_{\textrm{max}}}\odot A_q}_{\textbf{0}})\otimes Z_{ij}) -\underbrace{......}_{\textbf{Other zero matrices}})^{-1}x}
\end{multline*}

Notice that each of the 3 terms inside the summation are matrix blocks:\begin{itemize}
    \item Each $Y_{ij}\otimes Z_{ij}$ is a $t(t-1)$ by $t(t-1)$ matrix block in the top left-hand corner invariant of the candidate $q$.
    \item Each $B_{ij}\otimes Z_{ij}$ is a $t(t-1)$ by $t$ matrix block in the top right-hand corner.
    \item Each $B^T_{ij}\otimes Z_{ij}$ is a $t$ by $t(t-1)$ matrix block in the bottom left-hand corner.
\end{itemize}

Therefore, the matrix for which we wish to compute the inverse can be rewritten as a matrix in block form:

\begin{equation*}
\medmath{\begin{pmatrix}
I-c(\sum^l_{i=1}\sum^l_{j=i}Y_{ij}\otimes Z_{ij}) & -c(\sum^l_{i=1}\sum^l_{j=i}B_{ij}\otimes Z_{ij})\\
-c(\sum^l_{i=1}\sum^l_{j=i}B^T_{ij}\otimes Z_{ij}) & I
\end{pmatrix}}
\end{equation*}

As the largest block in the top left-hand corner is invariant of the candidate we evaluate, we can now speed up the kernel computation for subsequent candidates using blockwise matrix inversion, and present the control flow in Algorithm 1:

\begin{algorithm}
\caption{FastKernel: Single Member Replacement \label{lm:fastkernel}}
\SetKwInOut{Input}{Input}
\SetKwInOut{Output}{Output}
\Input{(1) The social network $G := \{A, L\}$; Lg, (2)
team members $T$ , (3) remaining team $R$, (4) Skill relevance matrix $W$}
\Output{Most suitable person $q$ to insert into $R$}
Initialize $A_0$, $A_R$ = adjacency matrix of $T$, $R$;\\
Initialize $L_0$, $L_R$ = skill matrix of $T$, $R$;\\
Precompute $Z_{ij} = W[i,j]*E_{0_{\textrm{max}(i,j)}}\odot A_0$ for all suitable $i, j$;\\
Precompute $Y_{ij} = E_{R_{\textrm{max}(i,j)}}\odot A_R$ for all suitable $i, j$;\\
Precompute $K^{-1} = (I-c(\sum^l_{i=1}\sum^l_{j=i}Y_{ij}\otimes Z_{ij}))^{-1}$;\\
Initialize $H = I$ with dimension $t$ by $t$;\\
\For{each candidate $q'$ in G after pruning }{
Initialize $A_{q'}, L_{q'}$;\\
Compute $B_{ij} = \textrm{max}(L_R[:,i]L_{q'}[:,j]^T, L_R[:,j]L_{q'}[:,i]^T)\odot A_q$ for all suitable $i, j$;\\
Compute $F = -c(\sum^l_{i=1}\sum^l_{j=i}B_{ij}\otimes Z_{ij})$;\\
Compute lower-right hand block $S = (H - F^TK^{-1}F)^{-1}$;\\
Compute upper-left hand block $K^{-1} + K^{-1}FSF^TK^{-1}$;\\
Compute upper-right hand block $-K^{-1}FS$;\\
Compute lower-left hand block $-SF^TK^{-1}$;\\
Combine blocks in lines 11-14 into inverse matrix $M$, compute candidate score of $q'$ as $y^TMx$;
}
\textbf{return} the candidate $q$ with the highest score.
\end{algorithm}

Note that in the single-member replacement case $R = T \setminus\{p\}$. We do not impose this as a constraint as we will have other uses for this algorithm in the remainder of the paper.

\begin{lemma}
\textsc{Time complexity of FastKernel.} When $R = T \setminus\{p\}$, Algorithm 1 has a time complexity of:

\begin{equation}
O(t^6 + l^2t^4 + (\sum_{i \in T \setminus \{p\}} \textnormal{degree}(i))(t^5 + l^2t^3))
\end{equation}
\end{lemma}

\begin{proof} 
The summation of Kronecker products in line 5 takes $O(l^2t^4)$ time and inverting $K$ takes $O(t^6)$ time. Inside the loop, computing $F$ on line 10 takes $O(l^2t^3)$ time, and computing Sherman-Morrison block form from lines 11 to 14 take $O(t^5)$ time.

\end{proof}

\subsection{REFORM - Simultaneous Multiple Member Replacement}

Armed with the single-member replacement algorithm, we now formulate our proposed solution for \textit{Subteam Replacement}. A straightforward way would be to initialize the remaining team $R = T \setminus S$ and iteratively add new team members to $R$ until $|R| = |T|$. As it is a well-known fact that greedy algorithms typically yield locally optimal solutions, we will need to ensure that a solution obtained greedily has a desirable lower bound on its quality. To do so, we will need to convert our graph kernel into a scoring function with some desirable properties for a greedy approach. We propose the following scoring function for a candidate subteam $S'$:

\begin{equation}
g(S') = \widehat{\textnormal{Ker}}(G_T, G_{R \cup S'}) - \widehat{\textnormal{Ker}}(G_T, G_R)
\end{equation}

Where $\widehat{\textnormal{Ker}}$ is defined as an approximate kernel measure:

\begin{equation}
\widehat{\textnormal{Ker}}(G_0, G_1) = \frac{\textnormal{sum}((I-c(E_{0x1}\odot(A_0\otimes A_1)))^{-1})}{|V_0|^4}
\end{equation}

Where the \textit{sum} operator returns the sum of all elements in the matrix. One can visualize the approximate kernel for cases where $|V_0| \geq |V_1|$ as appending `dummy nodes' with degree 0 and zero skill vectors to $G_1$: these dummy nodes are used for the sole purpose of making $G_1$'s size equal to $G_0$ before computing the graph kernel. In our case of using uniform starting and stopping vectors in the computation of the actual kernel value, the approximate kernel value is an underestimation: $\widehat{\textnormal{Ker}}(G_0, G_1) \leq \textnormal{Ker}(G_0, G_1) if |V_0| \geq |V_1|$.

\begin{lemma}
For any given team $T$ and 2 candidate teams with the same size $S', S'_2$, $\textnormal{Ker}(G_T, G_{S'}) \geq \textnormal{Ker}(G_T, G_{S'_2}) \Longleftrightarrow \widehat{\textnormal{Ker}}(G_T, G_{S'}) \geq \widehat{\textnormal{Ker}}(G_T, G_{S'_2}) \Longleftrightarrow g(S') \geq g(S'_2).$ 
\end{lemma}

\begin{proof}
Omitted for brevity.
\end{proof}

The above lemma implies that maximizing the actual kernel via \textsc{FastKernel} produces the same results as maximizing the approximate kernel and maximizing the scoring function. The purpose of the approximate kernel is to ensure that the scoring function defined with it holds the following desired properties: 

\begin{lemma}
$g(X)$ is normalized: $g(\emptyset) = 0$
\end{lemma}

\begin{proof}
Omitted for brevity.
\end{proof}

\begin{lemma}
$g(X)$ is non-decreasing: for every $X \subset V \setminus T$, $x \in V\setminus T$, $g(X\cup \{x\}) \geq g(X)$.
\end{lemma}

\textsc{Sketch of Proof.} Replacing a dummy node with and no connections and a zero skill vector with an actual person from the social network will have a non-negative impact on the kernel value. \QEDB

\begin{lemma}\label{lm:supermodular}
$g(X)$ is supermodular: for every $A \subset B \subset V \setminus T$, $x \in \setminus T$, $g(B\cup\{x\}) - g(B) \geq g(A\cup\{x\}) - g(A)$.
\end{lemma}

\begin{proof}
See Appendix.
\end{proof}

As our scoring function is normalized, non-decreasing and supermodular, it is a supermodular set function. We will now define a prototype greedy algorithm as in Algorithm 2:

\begin{algorithm}
\caption{GreedyMax}
\SetKwInOut{Input}{Input}
\SetKwInOut{Output}{Output}
\Input{(1) The social network $G := \{A, L\}$; Lg, (2)
team members $T$ , (3) subteam to replace $S$, (4) Skill relevance matrix $W$}
\Output{Most suitable subteam $S'$ to insert into $R = T \setminus S$}
Initialize $G_T$, $S' = \emptyset$, $C = V \setminus T$;\\
\While{$|S'| < |S|$}{
$v \leftarrow \textrm{max}_{v \in C} \quad g(S' \cup \{v\}) - g(S')$\;
$S' \leftarrow S' \cup \{v\}$\;
$C \leftarrow C \setminus \{v\}$\;
}
\textbf{return} $S'$.
\end{algorithm}

Let $S^+$ be the optimal solution, and $S'$ be the solution obtained via \textsc{GreedyMax}; as $g$ is a supermodular set function, we can bound the score of $S'$ in terms of the score of $S^+$:

\begin{equation}
g(S') \geq (1 - \kappa^g)g(S^+)
\end{equation}
\noindent where $\kappa^g$ is the `supermodular curvature' first introduced in \cite{SupermodularCurvature} as a dual to the submodular curvature. The supermodular curvature $\kappa^g$ is defined as follows:

\begin{equation}
\kappa^g = 1 - \textrm{min}_{v \in C} \frac{g(\{v\})}{g(C) - g(C \setminus \{v\})}
\end{equation}
\noindent where $C \subset V \setminus T$ is the set of candidates, namely the set of individuals that we can consider adding into $S'$.

Putting everything together, we now present \textsc{REFORM}, our greedy approximation algorithm which utilizes \textsc{FastKernel} and our graph kernel:

\begin{algorithm}
\caption{REFORM: Multiple member (Subteam) Replacement}
\SetKwInOut{Input}{Input}
\SetKwInOut{Output}{Output}
\Input{(1) The social network $G := \{A, L\}$; Lg, (2)
team members $T$ , (3) subteam to replace $S$, (4) Skill relevance matrix $W$}
\Output{Most suitable subteam $S'$ to insert into $R = T \setminus S$}
Initialize $S' = \emptyset$, $R = T \setminus S$;\\
\While{$|S'| < |S|$}{
$v \leftarrow $\textsc{FastKernel}$(G, T, R, W)$;\\
$S' \leftarrow S' \cup \{v\}$;\\
$R \leftarrow R \cup \{v\}$\;
}
\textbf{return} $S'$.

\end{algorithm}

\begin{lemma}

\textsc{Time complexity of REFORM.} Algorithm 3 has a time complexity of:
\begin{equation}
O(s(t^6 + l^2t^4 + (\sum_{i \in (R \cup S')} \textnormal{degree}(i))(t^5 + l^2t^3)))
\end{equation}
\end{lemma}
\begin{proof}
Omitted for brevity.
\end{proof}

\section{Experimental Evaluation}
\begin{table*}[ht]
\centering
\begin{tabular}{p{0.3\linewidth} p{0.3\linewidth} p{0.3\linewidth}}
    \includegraphics[width=\linewidth]{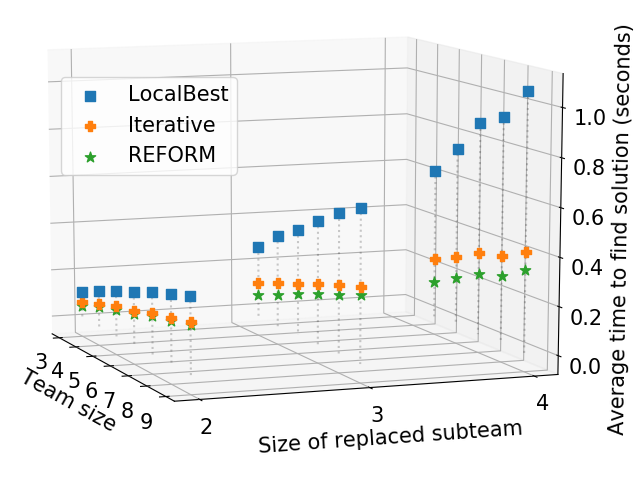}    &  \hspace{-0.3cm} \includegraphics[width=\linewidth]{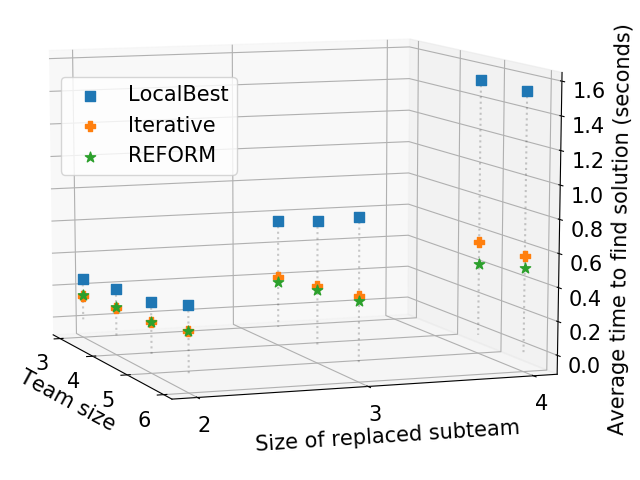}      \hspace{-0.3cm}        &   \includegraphics[width=\linewidth]{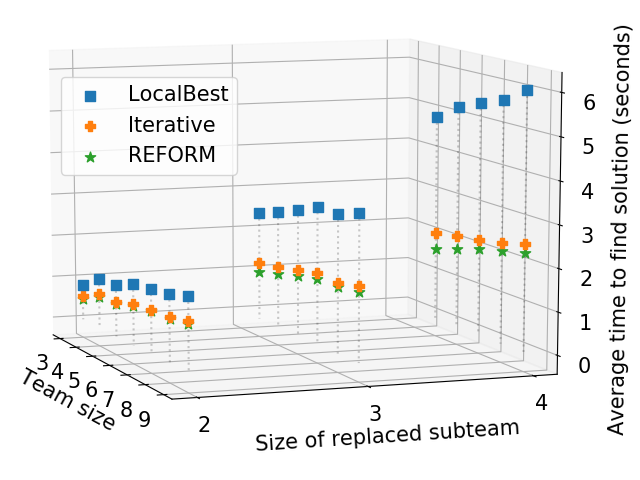}             \\
    \textbf{Fig. 2.} Average time for the algorithms to find the optimal solution vs. $t$ and $s$ on Synthetic BA graphs.
        & \textbf{Fig. 3.} Average time for the algorithms to find the optimal solution vs. $t$ and $s$ on DBLP subgraphs.
        & \textbf{Fig. 4.} Average time for the algorithms to find the optimal solution vs. $t$ and $s$ on IMDB subgraphs.\\
    \includegraphics[width=\linewidth]{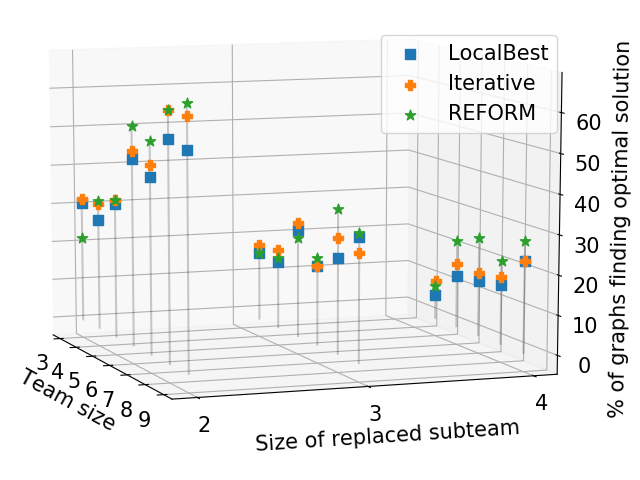}
                    &    \hspace{-0.3cm}   \includegraphics[width=\linewidth]{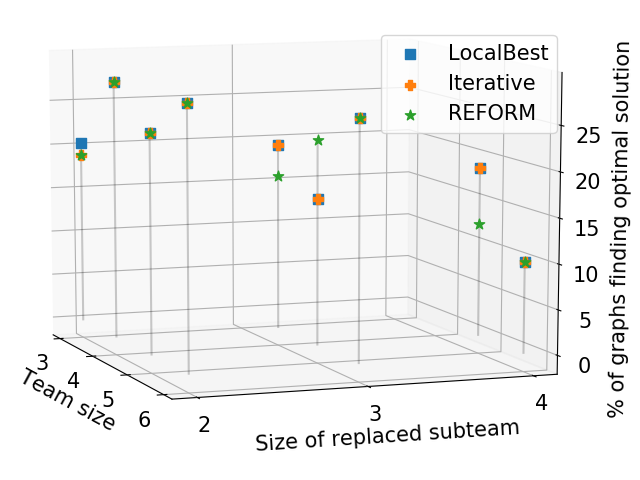}     \hspace{-0.3cm}    &         \includegraphics[width=\linewidth]{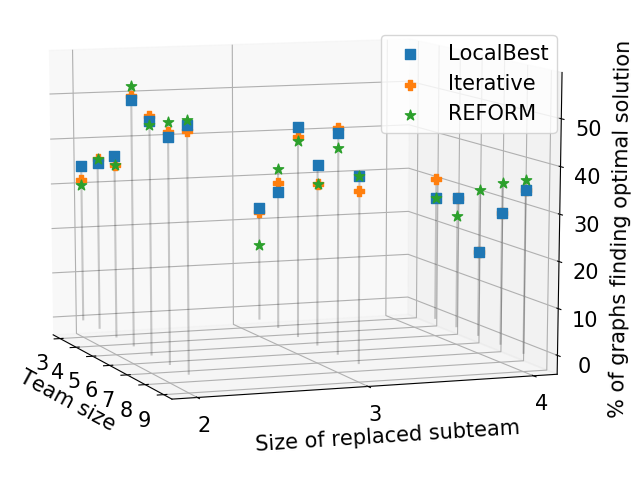}   \\
    \textbf{Fig. 5.} Percentage of graphs where the algorithms find the optimal solution vs. $t$ and $s$ on Synthetic BA graphs.
        & \textbf{Fig. 6.} Percentage of graphs where the algorithms find the optimal solution vs. $t$ and $s$ on DBLP subgraphs.
        & \textbf{Fig. 7.} Percentage of graphs where the algorithms find the optimal solution vs. $t$ and $s$ on IMDB subgraphs.\\
\end{tabular}
\end{table*}

In this section, we present the experimental results of our method and discuss their implications. We design our experiments according to the following questions which we wish to investigate: \begin{itemize}
    \item \textit{How viable is our proposed algorithm?}: We test our greedy algorithm against other baseline algorithms and evaluate its effectiveness in terms of both speed and solution quality.
    \item \textit{How effective are our metrics and recommendations?} We conduct a user study by running our solution along with other previous proposed solutions on teams in real-world datasets and use user-assigned rankings to assess the quality of the resulting replacement subteams.
\end{itemize}
\subsection{Datasets}
\textit{DBLP.} The DBLP dataset\footnote{http://dblp.uni-trier.de/xml/} contains information on papers published in different major computer science venues. We build the graph by creating a node for each author and edges between 2 authors with the weight as the number of papers they co-authored. We only consider the \textit{Inproceedings} file. The constructed network has $n = 989,686
$ nodes and  $m = 3,879,508
$ edges. For the skill matrix, we have $l= 19$ skills, one for each top venue in major Computer Science fields (e.g. WWW, KDD, NeuIPS, CVPR, AAAI), and the skill vector of an author is set to the number of papers they publish in each selected conference.

\textit{IMDB.} The IMDB dataset\footnote{http://grouplens.org/datasets/hetrec-2011/} contains information of movies and actors. We build the graph by creating a node for each actor and edges between 2 actors/actresses with the weight as the number of movies they co-starred. We only consider actors/actresses and movies from the U.S.. The constructed network has $n = 68,819$ nodes and $m = 2,911,749$ edges. For the skill matrix, we have $l = 20$ skills, one for each listed genre, and the skill vector is calculated via an exponential decay scoring system to address that the cast list tends to list actors/actresses in order of importance: Actors/actresses receive $0.95^{k-1}$ points in the genres of the movie they star in when their name is listed the $k^\textrm{th}$ on the cast list for that movie.

All experiments were run on a Linux machine with 8GB memory and an Intel i5-6200U CPU. The code will be released in the author's Github repository upon the publication of the paper.

\subsection{Quantitative evaluation}
We first demonstrate the viability of our greedy algorithm by gathering certain aggregate metrics from running the algorithm on batches of similarly-constructed graphs. Each batch consists of $100$ graphs, and within each batch, the team size $t$ and replaced subteam size $s$ is kept constant to investigate the relationship between our algorithm's performance and these properties. The skill relevance matrix $W$ is set as an upper triangular matrix of ones to weigh all skill pairing equally. We generate our batches of graphs via methods described below:

\textit{Synthetic BA graphs.} The Barabasi-Albert (BA) model is suitable as its preferential attachment property well mimics the structure of real social networks. We set the size of each graph to be fixed as $n = 50$, and the connectivity parameter to be $3$. Each node is labeled using $l = 6$ skills; we randomly sample the skill levels of each node and edge weights from an exponential distribution with $\lambda = 1$.

\textit{DBLP subgraphs.} We create DBLP subgraphs by specifying a certain year and 2 conferences, then taking the subgraph containing all authors with at least 1 publication in each of these 2 conferences in the specified year. This gives us subgraphs of interdisciplinary experts who are updated with the latest developments in their fields.

\textit{IMDB subgraphs.} We create IMDB subgraphs by specifying a certain year range and 2 genres, then taking the subgraph containing all actors/actresses starring in at least 1 movie in each of these 2 genres in the specified year range. Similarly to the DBLP subgraphs, this gives us subgraphs of multi-faceted actors/actresses who have recently shown their skill in performing in the specified genres.

In each synthetic BA graph, we randomly sample a connected subgraph to be the input team $T$, while we take a clique of authors/actors/actresses who have co-published a paper/co-starred in a movie to be the input team $T$ in the DBLP and IMDB subgraphs. In all teams, we select the subteam to replace $S$ randomly.

\begin{table}[ht]
\caption*{Table 3: Summary of properties of graphs used in batches for qualitative experimentation}
\begin{tabular}{llll}
\toprule
Dataset  & Graph size & Density & range of $t$ and $s$\\ \midrule
BA    &  50-50             &  0.115-0.115 &  $2\leq s \leq 4$, $s+1 \leq t \leq 9$        \\
DBLP &  30-100         &   0.016-0.106  &  $2\leq s \leq 4$, $s+1 \leq t \leq 6$\footnotemark[4]     \\
IMDB &  50-70         &   0.100-0.278\footnotemark[5]   &   $2\leq s \leq 4$, $s+1 \leq t \leq 9$   \\ \bottomrule
\end{tabular}
\textrm{\footnotemark[4]Papers coauthored by 7 or more researchers are too rare to form batches of sufficient size for experimentation.\\}
\textrm{\footnotemark[5]We avoid using subgraphs too dense as high density indicates a lack of diversity due to the subgraph containing actors/actresses from only a few distinct movies.}
\end{table}

To provide a comparison, we present the following alternative choices:\begin{itemize}
    \item \textsc{Iterative}: At each step, randomly select a person in $S$ that has yet to be replaced and replace them via \textsc{FastKernel} (Algorithm~\ref{lm:fastkernel}). This is repeated until all people in $S$ have been replaced.
    \item \textsc{LocalBest}: At each step, evaluate each person in $S$ that has yet to be replaced via \textsc{FastKernel}. The potential replacement that causes the largest increase (or least decrease) in kernel value is performed. This is repeated until all people in $S$ have been replaced.
\end{itemize}
We also employ a \textsc{Brute Force} algorithm to find the optimal solution $S^+$ by iterating through all possible subteams, which we use to compute the theoretical lower bound. We do not consider this algorithm as a baseline as it is computationally unfeasible in non-trivial problem settings.

\textit{Time to obtain solution.} Figures 2, 3, and 4 show that the proposed \textsc{REFORM} outperforms the baselines for any given combination of $t$ and $s$ on each dataset. Notably, all algorithms exhibit linear scaling with outgoing edges from $R$ which is roughly modeled by the increase in $t$, while \textsc{REFORM} and \textsc{Iterative} scales linearly with $s$ compared to \textsc{LocalBest}'s super-linear scaling with $s$. Comparing the results from the BA graphs and the DBLP subgraphs show that \textsc{REFORM} scales with the number of nodes in the graph, and comparing the results from the DBLP and IMDB subgraphs show that \textsc{REFORM} scales with the density of the graph.

\textit{Ability to find the optimal solution.} Figures 5, 6, and 7 showcase the computational difficulty of finding the optimal solution for \textit{Subteam Replacement}. No algorithm can find the optimal solution consistently, and all 3 algorithms have roughly equal capability in finding the optimal solution. However, a problem setting where less of the team is replaced (i.e. small$\frac{s}{t}$) generally translates to a higher probability for the algorithms to find the optimal solution.

\begin{figure}[htbp]
\centerline{\includegraphics[width=0.8\linewidth]{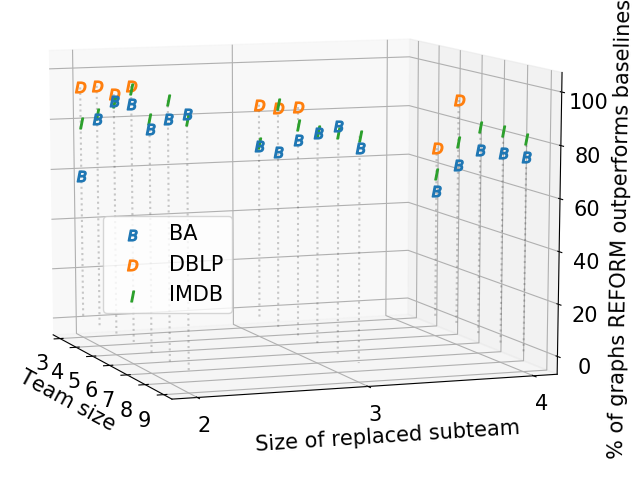}}
\captionsetup{labelformat=empty}
\caption{\textbf{Fig. 8.} Percentage of graphs where REFORM outperforms baselines vs. $t$ and $s$ on the three datasets.}
\end{figure}

\textit{Ability to outperform baselines.} Combining the information in Figures 2-7, we define that the proposed \textsc{REFORM} outperforms the baselines on a graph if it (1) finds a better solution or (2) finds the same solution but does it faster. It is evident in Fig. 8 that \textsc{REFORM} consistently outperforms the baselines in all the datasets and problem settings we experimented with, with this proportion being higher in problem settings where less of the team is replaced (small$\frac{s}{t}$). Interestingly, \textsc{REFORM} performs even better on the real datasets compared to the BA graphs.

\subsection{User studies}
We conduct user studies to evaluate our solution against solutions formed using established team performance metrics (TeamRep\cite{TeamReplacement}, GenDeR\cite{gender}) and 2 ablated versions (skill only, connection only) of our solution. 10 movies of various genre and year range were selected from the IMDB dataset, with the 4-6 person lead cast as the team and the 2-3 starring actors/actresses as the replaced subteam. The team, replaced subteam, and replacement subteam from the solutions are presented anonymously to the users in our visualization software (Fig. 9) with each actor/actresses represented by number of movies they participate in each genre. The order of presented solutions is also shuffled for each movie.

\begin{figure}[htbp]
\centerline{\includegraphics[width=\linewidth]{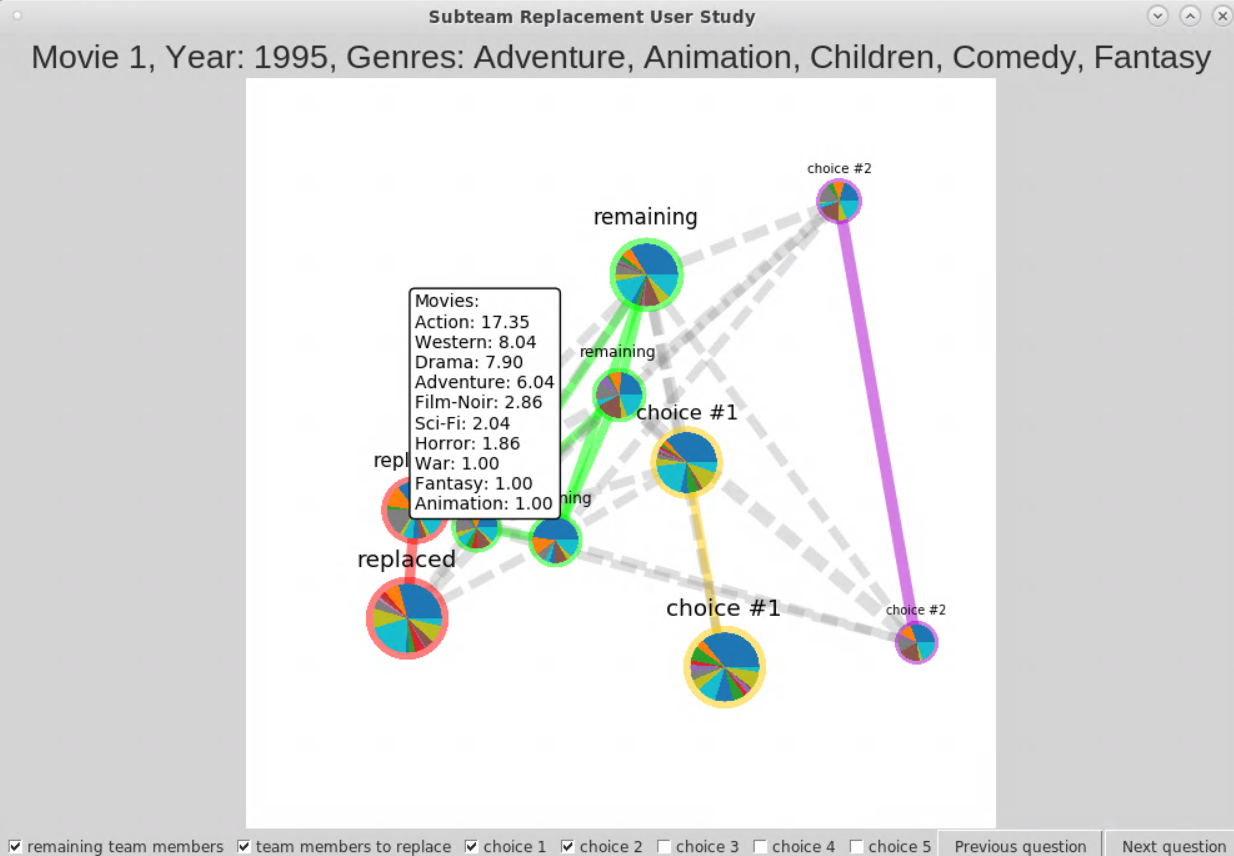}}
\captionsetup{labelformat=empty}
\caption{\textbf{Fig. 9.} Visualization software used in user study. Actors/actresses are anonymized and represented by their starred movies.}
\end{figure}

A total of 21 users participated in the study. For each question, users are asked to (1) select the first and second choice replacement subteams for each movie and (2) assign an ordinal confidence score (low, medium, high) which scales the weight of their answer for the particular question. The results are presented in Fig 10-12: 

\textit{Total score.} A method chosen as first and second is assigned 1 and 0.5 points respectively, scaled by the confidence score (high = 1, medium = 0.5, low = 0), and the total score is tallied from the choices of all users across all 10 questions. Fig. 11 shows that our method outperforms both established methods with statistical significance; it is outperformed by the `connection only' ablated method due to the relative ease of distinguishing connectivity differences compared to node attribute differences in the visualization.

\textit{NDCG (Normalized Discounted cumulative gain).} For each question, the NDCG for each method is computed from (1) the points awarded to each replacement subteams as relevance scores, and (2) the rank assigned to replacement subteams by the metrics employed in each method. Fig. 12 shows that our metric outperforms all baseline metrics except GenDeR with statistical significance; however, GenDeR is difficult to optimize for evident in its poor performance in terms of recommending replacement subteams.

\begin{table*}[ht]
\centering
\begin{tabular}{p{0.3\linewidth} p{0.3\linewidth} p{0.3\linewidth}}
    \includegraphics[width=\linewidth]{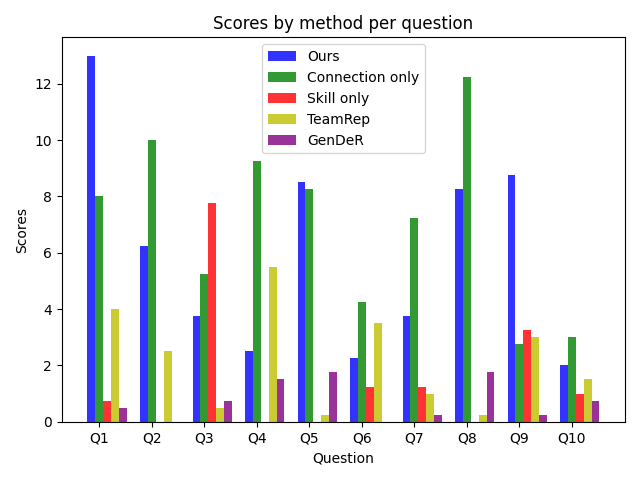}    &   \includegraphics[width=\linewidth]{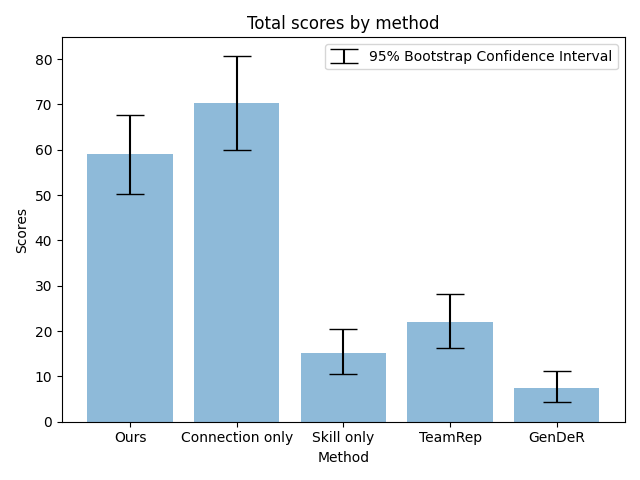} &  \includegraphics[width=\linewidth]{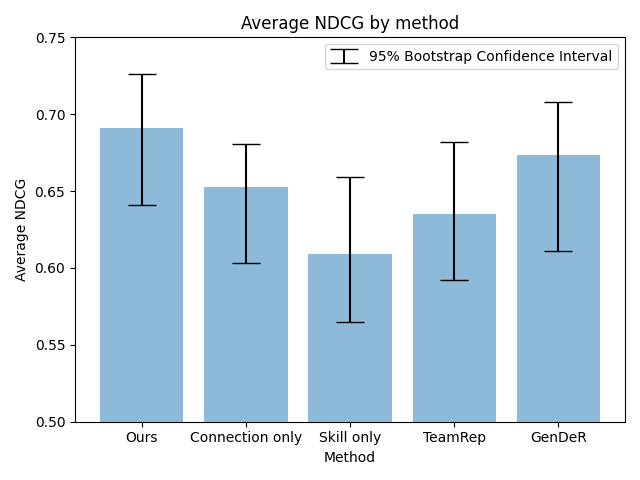} \\
    \textbf{Fig. 10.} Score per question for user study.
        & \textbf{Fig. 11.} Total scores for user study. Higher represents better replacement subteam. & \textbf{Fig. 12.} Average NDCG for user study. Higher NDCG represents more optimal ranking.\\
\end{tabular}
\end{table*}

\section{Related work}

\textbf{Team Formation.} Team formation is a problem concerning the assembly of individuals covering the skills required for given tasks under constraints such as team size or budget. Early works have modeled Team formation as a matching problem focusing on the skills present in the team\cite{zzkarian1999forming}\cite{baykasoglu2007project}. More recent works have taken the connectivity of team members into consideration and have used both existing \cite{cocktailparty}\cite{yin2018social} and constructed \cite{nearoptimalexpertteam} social networks as the problem setting, aiming to minimize communication cost alongside fulfilling the skill cover\cite{TeamFormation2009}, and have extended the problem to forming a sequence of teams given a sequence of skill covers \cite{OnlineTeamFormation}\cite{communitysystem}. A closely related problem is \textsc{Team replacement}, proposed in\cite{TeamReplacement} aiming to find a suitable replacement for a single unavailable member in a predefined team. Our work proposes an extension of \textsc{Team replacement} in the other direction of replacing multiple members at once while addressing several shortcomings in \cite{TeamReplacement}. 

\textbf{Team Performance Metrics.} 
Factors behind the success of teams have been identified and studied in social science literature, such as the familiarity between team members\cite{hinds2000choosing}\cite{familiarTeam} and the utilization of atypical combinations of skills \cite{AtypicalCombinations}, and have seen usage in gauging the success of a given team in areas such as scientific research \cite{PredictionAcademia} and online multiplayer games\cite{LolTeam}. Metrics used in team formation have involved properties derived from these factors, such as submodularity modelling diminishing returns\cite{submodularityteamformation} and distance encouraging diversity among team members\cite{gender}. Our work uses supermodularity in our goodness measure to model beneficial interaction within the team which we show to produce satisfactory results in our user study.

\textbf{Graph Kernels.}
Graph kernels are functions for measuring the similarity between graphs. Graph kernels can be roughly divided into a few families\cite{GraphKernelSurvey}: (1) Optimal assignment kernels, where similarity is defined by the best matching of nodes and/or edges between 2 graphs\cite{OptimalAssignmentVE}, (2) Subgraph pattern-based kernels, which uses the number of occurrences of pre-defined graphlets to define similarity \cite{GraphletKernel}, (3) Shortest-path kernels comparing the length of shortest paths between pairs of nodes in 2 graphs\cite{ShortestPath}, and (4) Random walk graph kernels, involving computing random walks on the product graph extendable to both node-labeled \cite{UkangKernel} and edge-labeled graphs \cite{SvnKernel}. A well-explored topic for random walk graph kernels is addressing its computational cost, as a direct kernel computation has a prohibitively high cost of $O(n^6)$. Speedup methods include speeding up individual kernel computations via methods such as Sylvester equations \cite{FastRWGK}\cite{SylvesterRWGK}, and reducing the total number of kernel computations needed when pairwise comparing numerous similar graphs\cite{TeamReplacement}. The latter speedup method has been shown to be achievable on node-labeled graphs kernels\cite{TeamReplacement}; In our work, we show that such strategies are also applicable to edge-labeled graph kernels.

\textbf{Best Effort Subgraph Matching.} Best Effort Subgraph Matching is the problem of finding the most similar subgraph to a given query graph in a larger graph. It is known that the problem is NP-Complete\cite{GraphSearchMultiAttribute}. A variety of applications have been recognized for subgraph matching such as Bioinformatics\cite{BioSubgrapMatching} and Team Formation\cite{du2017first}. Different problem formulations exist for subgraph matching: \cite{Tong2007FastBP} operates on weighted, attributed graphs (WAG) where each node has one attribute, while \cite{MAGE} operates on heterogeneous graphs. Our proposed method can be reformulated as a subgraph matching algorithm for WAGs with multi-attribute nodes using a query graph (the replaced team) as input.

\section{Conclusion}
In this paper, we introduce the problem of \textit{Subteam Replacement} to provide solutions for cases where a critical subteam in a large team becomes unavailable and needs to be replaced. To address the problem, we (1) formulate a novel random walk graph kernel which captures structural and skill similarity both locally and globally adaptable across different problem settings, (2) design a fast algorithm for the single-member replacement case, (3) build upon this algorithm to create our proposed solution \textsc{REFORM} for  \textit{Subteam Replacement}, and (4) prove that \textsc{REFORM} has a strict theoretical guarantee on its performance. Through experimenting with both synthetic and real datasets, we show that \textsc{REFORM} outperforms alternative choices both in terms of efficiency and quality and recommends more human-agreeable candidates compared to previous methods.

\appendix

\begin{figure*}[ht]
   \centering
   \includegraphics[width=0.8\linewidth]{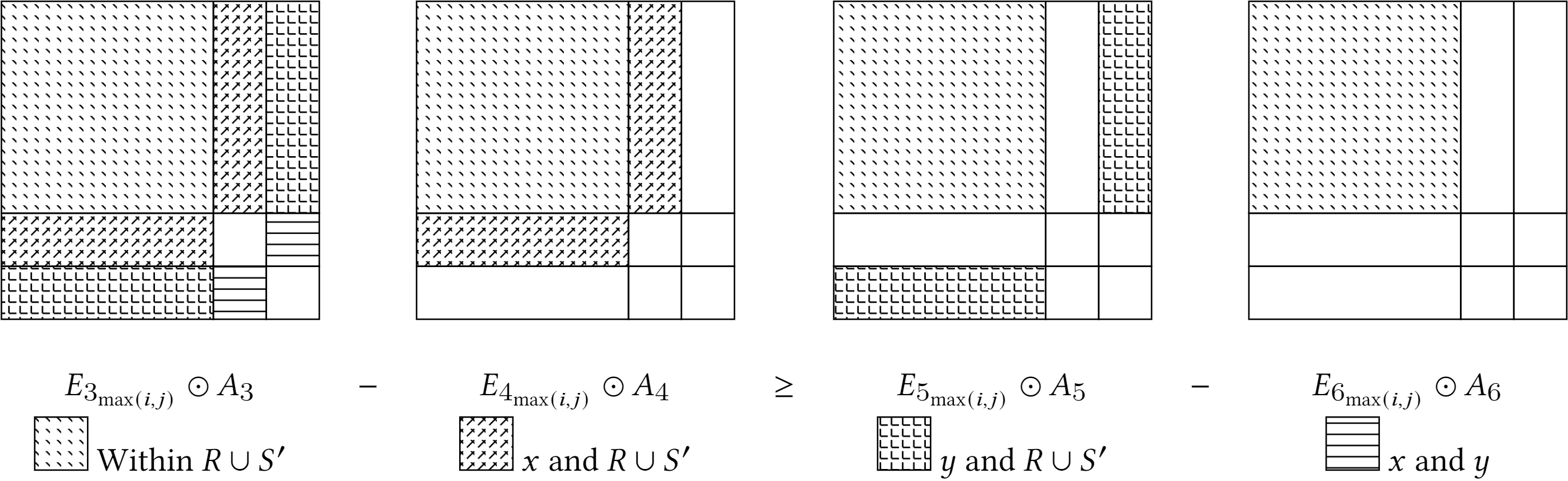}
   \captionsetup{labelformat=empty}
   \caption{\textbf{Fig. 13.} \textmd{A graphical illustration of the base case in Lemma 6}}
\end{figure*}

\textbf{Proof of Lemma~\ref{lm:supermodular}.}

\begin{proof} It is sufficient to show that for any $x, y \in V \setminus T$, $S' \subset V \setminus T$ the below inequality holds:
\begin{equation}
g(S'\cup\{x, y\}) - g(S'\cup\{x\}) \geq g(S'\cup\{y\}) - g(S')
\end{equation}
Which is equivalent to:
\begin{multline}
\widehat{\textnormal{Ker}}(G_T, G_{R\cup S'\cup\{x, y\}}) - \widehat{\textnormal{Ker}}(G_T, G_{R\cup S'\cup\{x\}}) \geq\\
\widehat{\textnormal{Ker}}(G_T, G_{R\cup S'\cup\{y\}}) - \widehat{\textnormal{Ker}}(G_T, G_{R\cup S'})
\end{multline}
For brevity, we will use the following symbols for the subgraphs:\begin{itemize}
    \item $G_3 = G_{R\cup S'\cup\{x, y\}} = \{A_{R\cup S'\cup\{x, y\}}, L_{R\cup S'\cup\{x, y\}}\}$
    \item $G_4 = G_{R\cup S'\cup\{x, o\}} = \{A_{R\cup S'\cup\{x, o\}}, L_{R\cup S'\cup\{x, o\}}\}$
    \item $G_5 = G_{R\cup S'\cup\{o, y\}} = \{A_{R\cup S'\cup\{o, y\}}, L_{R\cup S'\cup\{o, y\}}\}$ 
    \item $G_6 = G_{R\cup S'\cup\{o, o\}} = \{A_{R\cup S'\cup\{o, o\}}, L_{R\cup S'\cup\{o, o\}}\}$
\end{itemize}
Where $o$ is an isolated dummy node in the social network with a zero skill vector. It is trivial to show that $g(S'\cup\{y\}) = g(S'\cup\{o, y\})$ for $G_4$ and similarly for $G_5$ and $G_6$; the dummy nodes are added to ensure the size consistency between subgraphs. Without loss of generality, we assume the order of these nodes in the subgraph are as they are displayed in the above set of equations from left to right.
Using these notations the above inequality can be rewritten as:

\begin{multline}
\textnormal{sum}((I-G_{0\times3})^{-1}) - \textnormal{sum}((I-G_{0\times4})^{-1})\\
\geq \textnormal{sum}((I-G_{0\times5})^{-1}) - \textnormal{sum}((I-G_{0\times6})^{-1})
\end{multline}

Where $G_{\alpha \times \beta}$ is shorthand for $c(E_{\alpha \times \beta} \odot (A_\alpha \otimes A_\beta))$.

By Taylor expansion of the matrix inverse, it will therefore be sufficient to show that the following holds for all $k$:
\begin{equation}
G_{0\times3}^k - G_{0\times4}^k \geq G_{0\times5}^k - G_{0\times6}^k
\end{equation}

\textsc{Proof by induction - base case}: $k = 0$ is trivial, therefore we will start with $k = 1$:

\begin{equation}
G_{0\times3} - G_{0\times4} \geq G_{0\times5} - G_{0\times6}
\end{equation}

Fully expanding the above inequality, we get:

$\sum^l_{i=1}\sum^l_{j=i}cW[i,j]\!*\!(E_{0_{\textrm{max}(i,j)}}\!\otimes\! E_{3_{\textrm{max}(i,j)}})\!\odot\!(A_0\!\otimes\! A_3)$
$-\sum^l_{i=1}\sum^l_{j=i}cW[i,j]\!*\!(E_{0_{\textrm{max}(i,j)}}\!\otimes\! E_{4_{\textrm{max}(i,j)}})\!\odot\!(A_0\!\otimes\! A_4)$
$\geq\sum^l_{i=1}\sum^l_{j=i}cW[i,j]\!*\!(E_{0_{\textrm{max}(i,j)}}\!\otimes\! E_{5_{\textrm{max}(i,j)}})\!\odot\!(A_0\!\otimes\! A_5)$
$-\sum^l_{i=1}\sum^l_{j=i}cW[i,j]\!*\!(E_{0_{\textrm{max}(i,j)}}\!\otimes\! E_{6_{\textrm{max}(i,j)}})\!\odot\!(A_0\!\otimes\! A_6)$

By the mixed property of the Kronecker and Hadamard products, we rewrite the above as follows:

$\sum^l_{i=1}\sum^l_{j=i}cW[i,j]*(E_{0_{\textrm{max}(i,j)}}\odot A_0) \otimes(E_{3_{\textrm{max}(i,j)}}\odot A_3)$
$-\sum^l_{i=1}\sum^l_{j=i}cW[i,j]*(E_{0_{\textrm{max}(i,j)}}\odot A_0) \otimes(E_{4_{\textrm{max}(i,j)}}\odot A_4)$
$\geq\sum^l_{i=1}\sum^l_{j=i}cW[i,j]*(E_{0_{\textrm{max}(i,j)}}\odot A_0) \otimes(E_{5_{\textrm{max}(i,j)}}\odot A_5)$
$-\sum^l_{i=1}\sum^l_{j=i}cW[i,j]*(E_{0_{\textrm{max}(i,j)}}\odot A_0) \otimes(E_{6_{\textrm{max}(i,j)}}\odot A_6)$

It is evident that the base case holds given the below equation holds for all suitable $i, j$:
\begin{multline}
E_{3_{\textrm{max}(i,j)}}\odot A_3 - E_{4_{\textrm{max}(i,j)}}\odot A_4\\
\geq E_{5_{\textrm{max}(i,j)}}\odot A_5 - E_{6_{\textrm{max}(i,j)}}\odot A_6 
\end{multline}

As shown in Fig. 13, the only different elements in the matrices on the LHS and RHS are the two elements in the indices $(|R| + |S'| + 1, |R| + |S'| + 2)$ and $(|R| + |S'| + 2, |R| + |S'| + 1)$, which correspond to the adjacency between $x$ and $y$ scaled by the similarity of their skill vectors. We can confirm that these elements are $0$ in RHS and may be nonzero in LHS as both $x$ and $y$ are in $G_3$, therefore the base case holds.

\textsc{Proof by induction - inductive step}: assuming equation 12 holds for all $m \in \{1,...,k-1\}$, we will prove that equation 12 holds for $k$: 

\begin{multline*}
\medmath{\textnormal{LHS}: G_{0\times3}^k - G_{0\times4}^k} \\
\medmath{= (G_{0\times3}^{k-1} - G_{0\times4}^{k-1}) * (G_{0\times3} - G_{0\times4}) + G_{0\times3}^{k-1}G_{0\times4} + G_{0\times3}G_{0\times4}^{k-1}}
\end{multline*}

By the induction assumption, we get:

\begin{equation*}
\geq (G_{0\times5}^{k-1} - G_{0\times6}^{k-1}) * (G_{0\times5} - G_{0\times6}) + G_{0\times3}^{k-1}G_{0\times4} + G_{0\times3}G_{0\times4}^{k-1}
\end{equation*}

For our defined element-wise matrix inequality operator, we have that given any non-negative matrices $A, B, C, D$:\begin{itemize}
    \item $A \geq B \land C \geq D \implies AC \geq BD$
    \item  $A \geq B \implies A^k \geq B^k ,\forall k > 0$
\end{itemize}
We omit the proofs of the above 2 statements for brevity. Given the way the subgraphs are formed, we have $G_{0\times3} \geq G_{0\times5}$ and $G_{0\times4} \geq G_{0\times6}$ (see Fig. 2). Therefore, we get:
\begin{multline}
\medmath{\geq (G_{0\times5}^{k-1} - G_{0\times6}^{k-1}) * (G_{0\times5} - G_{0\times6}) + G_{0\times5}^{k-1}G_{0\times6} + G_{0\times5}G_{0\times6}^{k-1}}\\
\medmath{= G_{0\times5}^k - G_{0\times6}^k = \textnormal{RHS}}
\end{multline}

Thus finishing the proof.
\end{proof}

\bibliographystyle{IEEEtran}
\bibliography{bibli}
\end{document}